\documentclass[11pt]{article}
\usepackage[utf8]{inputenc}
\usepackage[T1]{fontenc}
\usepackage{lmodern}
\usepackage{geometry}
\usepackage{graphicx}
\usepackage[dvipsnames]{xcolor}
\usepackage[colorlinks=true]{hyperref}
\hypersetup{
    linkcolor=Magenta,
    citecolor=ForestGreen,
    urlcolor=Magenta
}
\usepackage{amsmath,amssymb,amsthm,amsfonts,latexsym,bbm,xspace,graphicx,float,mathtools,braket}
\usepackage{algorithm}
\usepackage{algpseudocode}
\usepackage{complexity}
\usepackage{appendix}
\usepackage{enumitem}
\usepackage[nameinlink,capitalize]{cleveref}  
\theoremstyle{plain}
\newtheorem{theorem}{Theorem}[section]
\newtheorem{lemma}[theorem]{Lemma}
\newtheorem{corollary}[theorem]{Corollary}
\newtheorem{proposition}[theorem]{Proposition}
\newtheorem{definition}[theorem]{Definition}
\renewcommand{\hat}[1]{\widehat{#1}}
\newcommand{\abs}[1]{\left|#1\right|}
\newcommand{\question}[1]{\par\medskip\noindent\textbf{Question.} \textit{#1}\par\medskip}
\DeclareMathOperator{\Tr}{Tr}

\renewcommand{\poly}{\operatorname{poly}}
\let\eps\varepsilon

\allowdisplaybreaks

\title{Learning Sparse Quantum States}
\author{
  Aniruddha Sen \\
  The University of Texas at Austin \\
  \href{mailto:aniruddhasen@utexas.edu}{\texttt{aniruddhasen@utexas.edu}}
}
\date{September 10, 2026}

\begin{document}

\maketitle

\begin{abstract}
We study the problem of tomography for $k$-sparse quantum states. In contrast to classical distribution learning, where tight sample and time complexity bounds in terms of support size are well understood, no non-trivial bounds were previously shown for this problem. We give the first near optimal algorithm for learning $n$-qubit $k$-sparse pure quantum states, obtaining fidelity at least $1-\eps$ with high probability using $\Tilde{O}(k/\eps)$ copies of the state and $\Tilde{O}(kn/\eps)$ time. Both bounds are optimal up to polylogarithmic factors. As an implication, we also obtain an algorithm with near optimal $\Tilde{O}(kr/\eps)$ sample complexity for learning $k$-sparse rank-$r$ mixed states, via the random purification channel technique. Obtaining time complexity nearly matching the sample complexity, for $r>1$, remains an important open question.
\end{abstract}

\section{Introduction}
Classical distribution learning algorithms are known to almost always scale with the sparsity for $k$-sparse distributions, since sampling from discrete distributions is naturally sensitive to the support \cite{kearns1994learnability, valiant2016instance, canonne2020short}. However, the quantum analogue of learning states seems to not immediately take advantage of such structure if it exists. Quantum state tomography is the problem of obtaining a classical description of a quantum state given multiple copies of it. We mainly focus on the tomography of pure quantum states, with the goal of obtaining sample and time complexity bounds which scale with the size of the support of the state in some fixed basis. We consider support in the computational basis throughout this paper, and define $k$-\textit{sparse} pure states as states of the form
\begin{equation}
\ket{\psi}=\sum_{x\in S}\alpha_x\ket{x},\qquad S\subseteq\{0,1\}^n,\quad |S| \le k
\end{equation}
So when $k=2^n$, we recover the definition of general pure states. Our main result is as follows.

\begin{theorem}\label{theorem:pure-main-theorem}
Given as input multiple copies of an unknown $n$-qubit $k$-sparse pure state $\ket{\psi}$, the support size $k$, and error parameters $\eps,\delta > 0$, there exists an algorithm (\cref{alg:k_sparse_pure_tomography}) that outputs a classical description of an estimate $\ket{\hat{\psi}}$ satisfying $|\bra{\psi}\hat{\psi}\rangle|^2 \ge 1 - \eps$ with probability at least $1-\delta$. The algorithm performs only single-copy measurements and requires $\Tilde{O}(k/\eps)$ samples and $\Tilde{O}(kn/\eps)$ time, when setting $\delta=1/n$, where the $\Tilde{O}$ suppresses logarithmic factors in $k,n$ and $1/\eps$.
\end{theorem}

The main technical contribution of this work is to obtain a near-linear scaling in the support size $k$ for the gate or \textit{time complexity} of our pure state tomography algorithm. Note that general pure state tomography takes $\Theta(2^n/\eps)$ samples and time, the latter up to log factors, for learning with infidelity $\eps$ \cite{o2016efficient,haah2016sample,grewal2026nearly}. However, these complexity bounds are independent of the support size $k$ of the state, and hold for worst case $k$ which is when the state has full support. When $k=2^n$, we recover the optimal worst-case sample and time complexity up to log factors; but for $k\ll2^n$, say $k=2^{n/2}$ or $k=n^2$, our algorithm performs exponentially better than these algorithms. 

We also consider $k$-sparse mixed states and define them analogously to pure states. A mixed state $\rho$ has sparsity $k$ iff the probability distribution obtained by measuring $\rho$ in the computational basis has support size at most $k$. Alternately, $\rho$ is $k$-sparse if there exists a set $S\subseteq\{0,1\}^n$ with $|S| \le k$ such that
\begin{equation}
\rho=\Pi_S\rho\Pi_S,
    \qquad
\Pi_S\coloneqq\sum_{x\in S}|x\rangle\langle x|.
\end{equation}
Note that this is distinct from the \textit{rank} $r$ of a mixed state, though we can always bound $r \le k$. So, $k$-sparse pure states are a special class of $k$-sparse mixed states. Additionally, diagonal mixed states with $k$ non-zero entries correspond directly to classical distributions over pure basis states, and are also $k$-sparse by this definition. In a practical setting, $k$-sparse pure states can often be subjected to noise within their support. This motivates the following theorem for learning $k$-sparse mixed states, the proof of which goes through the recently introduced random purification channel technique \cite{tang2025conjugate, girardi2025random}.

\begin{theorem}\label{theorem:mixed-main-theorem}
Given as input multiple copies of an unknown $n$-qubit $k$-sparse rank-$r$ mixed state $\rho$, the support size $k$, the rank $r$, and error parameters $\eps,\delta > 0$, there exists an algorithm (\cref{alg:k_sparse_mixed_tomography}) that outputs a classical description of an estimate $\hat{\rho}$ satisfying $F(\rho,\hat{\rho}) \footnote{Here $F(\rho,\hat{\rho}) \coloneqq \left(\Tr\sqrt{\sqrt{\rho}\,\hat{\rho}\,\sqrt{\rho}}\right)^2$ denotes the fidelity between mixed states.} \ge 1-\eps$ with probability at least $1-\delta$. The algorithm requires $m=\Tilde{O}(kr/\eps)$ samples and $\poly(m,n)$ time, when setting $\delta=1/n$, where the $\Tilde{O}$ suppresses logarithmic factors in $k$ and $1/\eps$.
\end{theorem}

A simple lower bound via a reduction to mixed state tomography fidelity lower bounds shows that $\Omega(kr/\eps)$ samples, and consequently $\Omega(knr/\eps)$ time in order to read the $n$-qubit inputs of each sample, is required to learn $k$-sparse rank-$r$ states to fidelity at least $ 1-\eps$. We defer this proof to \cref{lemma:lower-bound} in the Appendix. This shows that our pure state tomography algorithm is both sample and time optimal in parameters $n,k,1/\eps$, up to log factors, and our mixed state tomography algorithm is nearly sample optimal in the same parameters, and the rank $r$, but is likely not time optimal. Indeed, obtaining a time optimal mixed state tomography algorithm, even one which may not scale with sparsity, remains an important open problem.

We now discuss some implications of our results, and connections to different themes in the literature.

\paragraph{Polynomial time learning.} Consider the class of states which are $k$-sparse in some \emph{efficiently} computable basis, with $k=\poly(n)$. A trivial example is the state $\ket{+}^{\otimes n}$, which is $2^n$-sparse in the Z basis but $1$-sparse in the X basis. Our results directly imply a \emph{polynomial time} learning algorithm (and in particular, potentially linear time depending on $k$ and the basis) for states in this class given we know the sparse basis, by simply applying the basis change operation on all copies. This is true for both pure and mixed states, since the rank $r \le k$ always, so $k=\poly(n) \implies r=\poly(n)$. This extends the line of work on polynomial time learning algorithms for special classes of quantum states such as stabilizer states \cite{montanaro2017learning} or near stabilizer states \cite{grewal2023efficient,grewal2024improved}, states with bounded stabilizer rank or extent \cite{arunachalam2026learning,arunachalam2026tomography}, non-interacting fermion states \cite{aaronson2021efficient}, and matrix product states \cite{cramer2010efficient}. Moreover, for pure states, our algorithm only uses CNOT circuits of low depth with size at most $O(n\log k)$ followed by single qubit Pauli measurements, making it potentially feasible to run on near-term quantum computers for states with small $n,k$.

Various commonly studied states are sparse with support size at most $\poly(n)$. With respect to sparsity in the computational basis, states such as the $n$-qubit GHZ (2-sparse) and $W$ states ($n$-sparse) are widely studied. We can also consider sparse vector encoding states arising in quantum algorithms for linear algebra and machine learning \cite{saeedi2019quantum,prakash2014quantum}; and in quantum chemistry \cite{feniou2024sparse}. Efficient algorithms for sparse state preparation have been widely studied \cite{feniou2024sparse, malvetti2021quantum, zhang2022quantum, mao2024toward}, giving another clear motivation for optimal learning algorithms. Sparsity also occurs naturally in many-body systems due to conservation laws, such as fixed particle number or magnetization, which restrict states to exponentially smaller subspaces. For example, states with exactly $r$ particles have support size $\binom{n}{r}$, which is polynomial in $n$ for constant $r$, and include Dicke states and few-particle Fock states \cite{dicke1954coherence, bruus2004many}. However, the corresponding physical basis may not always be efficiently transformable to the computational basis.

\paragraph{Scaling with rank for mixed state tomography.} Essentially all state-of-the-art sample and time complexity bounds for learning mixed quantum states scale with the rank of the density matrix corresponding to the state \cite{o2016efficient,haah2016sample,brandao2020fast,pelecanos2025mixed}. We propose sparsity as the main analogous quantity of interest in the special case of pure states, while also being well-defined for mixed states. Along these lines, our work shows that for pure states, we obtain a smooth scaling with the support size of the state. Indeed, we can obtain tight bounds for time complexity in terms of the support size, which matches the sample complexity upto a necessary factor of $n$. This is contrary to the case of mixed states, where our best time complexity bounds are at least polynomially worse than sample complexity in terms of $r$ and $d$, and are moreover not known or expected to be optimal \cite{brandao2020fast,pelecanos2025mixed}. 

\subsection{Prior Work}
The closest prior work for learning $k-$sparse states comes from general mixed state and pure state tomography algorithms \cite{o2016efficient,haah2016sample,grewal2026nearly,van2023quantum,pelecanos2025mixed}. However, analysis of these algorithms do not directly imply any scaling with the sparsity of the state. Following prior work for pure state tomography, it \emph{is} possible get a simple $\Tilde{O}(k/\eps)$ sample and $\Tilde{O}(nk^2/\eps)$ time algorithm for pure states by first learning the support of the state through computational basis measurements, applying a permutation which maps the support of the state onto the first $k$ basis states, and then running the algorithm by $\cite{van2023quantum}$ on it; a careful analysis gives us the claimed complexity bounds. A structurally similar algorithm was also given in $\cite{grewal2025efficient}$ which uses $O((k+\log(n))/\eps)$ samples and $O(k^3+nk^2/\eps)$ time. But this still does not get the near-linear running time in $k$ that we want. Rather, as a direct application of our result, we improve the time complexity of the Pauli-sparse unitary channel tomography algorithm of \cite{grewal2025efficient}, up to a logarithmic loss in error parameter dependence. Obtaining optimal time complexity bounds has often been a more challenging problem than showing optimal sample complexity for learning problems \cite{aaronson2018shadow,huang2020predicting,huang2024certifying}, since the latter is a purely information theoretic claim; in our case the main challenge was to obtain linear scaling in $k$ for the time complexity. 

Indeed, the optimal time complexity bound for pure state tomography (independent of $k$) was only recently resolved in \cite{grewal2026nearly}. Unfortunately, their algorithm does not get optimal scaling in samples or time with the support size $k$ due to a subroutine of Frobenius distance estimation which involves sampling exponentially many Pauli expectation values. This subroutine is independent of the support of the state in the computational basis. We complement their result by recovering their time complexity bounds up to $\log$ factors via entirely different techniques that scale with $k$, at the cost of using some two-qubit gates, and thus obtain an alternative proof of the optimal time complexity of pure state tomography as well. 

\subsection{Phase States}
It will be useful to understand our algorithm for the special case of phase states. We motivate this here and give some evidence for why they would also be an interesting case to study independently. Phase states are a subclass of pure states which are of the form $\ket{\psi} = \frac{1}{\sqrt{2^n}}\sum_{x\in\{0,1\}^n} e^{i\theta_x}\ket{x}$, and have applications in quantum cryptography and state synthesis \cite{ji2018pseudorandom, irani2021quantum, rosenthal2024efficient}. Several of the main ideas and techniques used in the algorithm and proof of the main result for learning $k$-sparse pure states appear simply when we analyze the special case of learning \emph{$k$-sparse phase states} which are states of the following form.
\begin{equation}
\ket{\psi} = \frac{1}{\sqrt{k}}\sum_{\substack{x\in S\subseteq\{0,1\}^n \\ |S|=k}} e^{i\theta_x}\ket{x} \,,\quad \theta_x \in [0,2\pi)
\end{equation}
In particular, when $k=2^n$ we recover the usual definition of phase states. Sparse phase states can also be viewed as a generalization of the previously studied \emph{subset phase states} in the context of pseudorandom states \cite{aaronson2022quantum} and \textit{degree-$d$ sparse phase states} in quantum learning theory \cite{arunachalam2022optimal}, where additional constraints were imposed on the algebraic structure of the phases.

An implication of our result for the special case of sparse phase states is that it shows that the security of a particular pseudorandom quantum state, referred to as a \textit{subset phase state} in \cite{aaronson2022quantum}, is almost tight in the size of the subset $k$. The security proof of their construction shows computational pseudorandomness only for states with support size $k \ge O(t^2/\eps)$ to get $\eps$-security (via statistical distance to a Haar random state) against at most $t$-\textit{copy} adversaries. Complementing this, our learning algorithm implies that even the \textit{time} complexity of learning this state to $\eps$-accuracy is at most $O(k\mathrm{polylog}(k)/\eps)$, distinguishing it from random states within at most this time.

More generally, we also show that the class of $\poly(n)$-sparse phase states can be learned near-optimally in polynomial time, which was posed as an open question in \cite{arunachalam2022optimal} for learning degree-$d$ phase states, we resolve the related case of general phase states.

\subsection{Technical Overview}
For ease of understanding and presentation, we start by giving a brief sketch of the algorithm and analysis for the special case of phase states, and later extend it with a single additional key idea to obtain our main result on pure states. Our mixed state result follows consequently.

\paragraph{Learning Phase States} Given copies of a $k$-sparse phase state, we first learn the support $S$ of the state such that $\abs{S}=k$ by measuring the state $\approx k\log(k)$ times in the computational basis, by a coupon collector argument (it is not necessary to be given $k$ as an input, but here we assume it for simplicity). The remaining task is to learn the phases on this known support. At a high level, we do this by projecting on to independent pairwise superpositions between several \emph{pairs} of basis states. This is done by constructing a set of hash functions $f$, which are roughly 2-to-1 on $S$; we apply this in superposition to the state, and then measure.

More concretely, we apply an XOR oracle according to $f$ to the state (obtaining a superposition over $\ket{x}\ket{f(x)}, x\in S$) and measure the function register, getting some outcome $y$. If $y$ has exactly two pre-images in $S$, then the post-measurement state is a superposition of the form $\frac{1}{\sqrt{2}}(\ket{x_1}+e^{i({\theta_{x_2}-\theta_{x_1}})}\ket{x_2})$ up to a global phase. In particular, the relative phase between $\ket{x_1}$ and $\ket{x_2}$ is preserved. With enough copies of this state, we can obtain multiple copies of each specific two state superposition and learn their relative phase up to some desired precision.

Now, consider the graph $G=(S,E)$ where the vertices are indexed by basis states, and the edges indicate the relative phase between the basis states at its endpoints. We show that it suffices to learn enough relative phases (or edges in this graph) so there is a `short' path (of length $O(\log(k))$) from some reference state $r$ to all the other basis states. Each function $f$ defines a partial matching in this graph, which we learn by projecting onto the states corresponding to the edges. Then, by taking the union over several such matchings we obtain a sufficiently well-connected graph. This allows us to learn the phases between the reference state and all other states efficiently. At this point, it just remains to specify this function $f$. 

It turns out \emph{uniform random} functions have the desired properties, but are too expensive to implement in a circuit. We need to use more structured low-complexity functions to get a near time-optimal algorithm. Consequently, we make another natural choice and we prove in \cref{lemma:linear-connectivity} that it suffices to pick random linear functions over $\mathbb{F}_2$, which gets us near optimal circuit size. This allows us to prove the following theorem.

\begin{theorem}\label{theorem:phase-main-theorem-intro}
Given as input multiple copies of an unknown $n$-qubit $k$-sparse phase state $\ket{\psi}$, the support size $k$, and error parameters $\eps,\delta > 0$, there exists an algorithm (\cref{alg:k_sparse_phase_tomography}) that outputs a classical description of an estimate $\ket{\hat{\psi}}$ satisfying $|\bra{\psi}\hat{\psi}\rangle|^2 \ge 1 - \eps$ with probability at least $1-\delta$. The algorithm performs only single-copy measurements and requires $\Tilde{O}(k/\eps)$ samples and $\Tilde{O}(kn/\eps)$ time, when setting $\delta=1/n$, where the $\Tilde{O}$ suppresses logarithmic factors in $k$ and $n$.
\end{theorem}

\paragraph{Learning Pure States} The same ideas from the phase state case also extend to arbitrary sparse pure states, where the only difference is that the probabilities in the $Z$ basis are no longer uniform. We can however learn the probability distribution in the $Z$ basis by simply using a generic classical distribution learning algorithm and generating samples through $Z$ measurements. The problem is that, once the probabilities are non-uniform, the phases of all the basis states should not be learned to the same accuracy, unlike phase states. A basis state with small probability mass should be learned less precisely, proportional to its contribution to the fidelity. 

On the other hand, if we build one global phase graph, then a low probability basis state might lie on a path between two high probability basis states, forcing us to learn an edge incident to a low probability state to a much higher accuracy than its own mass justifies. 

The key idea is that this can be handled by using the standard bucketing technique from classical distribution learning which says the following: It is possible to partition any discrete distribution into a small number of probability buckets, so that within each bucket all probabilities are approximately equal (up to a factor of 2 in our case). Thus, restricted to a single bucket, the state looks approximately like a phase state. We can then apply the same hashing construction with random linear functions as in the phase state algorithm, separately for each bucket. In order to make the phases learned in different levels be consistent with each other, we also add the same reference basis state $r$ to every level and learn all the phases relative to $r$. Summing the error over all buckets proportional to the probability mass of each gives the final fidelity error guarantee for arbitrary sparse pure states. It is worth noting that the sample and time complexity bounds for general sparse pure states match the case of sparse phase states, showing that the seemingly simpler case of learning sparse phase states is just as hard as the general case.

\paragraph{Learning Mixed States} We further extend our pure state learning algorithm to mixed states via the random purification channel construction introduced by \cite{tang2025conjugate}. In essence, their result shows that there exists an efficiently implementable quantum channel such that, given $m$ copies of a mixed state $\rho$, it outputs $m$ copies of a uniformly random purification of $\rho$. Given this channel, we show that $k$-sparse rank-$r$ mixed states get mapped to $kr$-sparse pure states. Then, we run our sparse pure state tomography algorithm on the purified samples, from which we can obtain a classical description of the mixed state to the desired accuracy, following the work of \cite{pelecanos2025mixed}. Our mixed state tomography algorithm is near sample optimal, but no longer time optimal since current constructions of the random purification channel rely on the Schur transform and are comparatively expensive to implement.

\paragraph{Robustness and input description.} Here we comment on a couple of important properties that our algorithms exhibit. Modern quantum systems are inherently noisy \cite{preskill2018quantum}, and it is often desirable for learning algorithms to be robust to small amounts of noise in a quantum system \cite{cotler2026noisy, hu2025demonstration}. We note that our algorithms are robust to small noise (say if our states are acted upon by a depolarizing noise channel with small probability) -- in particular we obtain the same sample complexity guarantees if the input state is approximately $k$-sparse i.e. if at least $1-\eps$ of the probability mass is on at most $k$ basis states, as opposed to exactly $k$-sparse where the probability mass needs to exactly 0 outside of the $k$-size support. This follows simply from the fact that the classical distribution learning algorithm we use in the first step of (say) the pure state case is a \textit{learned} support in any case, it does necessarily need to know the exact support in order to learn the state to the desired precision. A similar argument works for the mixed state (and phase state) case as well, by reduction to the pure state case. 

Moreover, our algorithms need not necessarily be given the sparsity parameter ( i.e. the support size) $k$ as input either. This might be useful if we know that our states are noisy or sparse, but we don't have a good upper bound on $k$. For the phase state case, this is particularly clear via the coupon collector argument, since we can simply choose to stop when we haven't seen a new coupon for at most double the number of draws we needed for the last coupon, with additional log factors in  $1/\delta$ with respect to our desired failure tolerance $\delta$. More generally, in the pure state (and consequently mixed state) case, we again appeal to classical distribution learning algorithms which have this property \cite{valiant2016instance,mazzetto2024improved}, and our previous argument similarly follows since we just need to learn the phases on the learned support, to obtain our fidelity guarantee.

Throughout the paper, we assume $k$ is given as input to the algorithm, and that the state is exactly $k$-sparse, for the sake of conciseness. However, as we have described here, this assumption is not necessary and all our results would remain the same with minor modifications.

\subsection{Open Questions}
As mentioned previously, the main open question from our work is to obtain optimal time complexity bounds for mixed states. Broadly, there exist very few time optimal (or nearly time optimal) bounds for tomography problems in the literature. We consider a few concrete variants which are open.

\question{Does there exist an algorithm for learning $n$-qubit rank-$r$ mixed states in $O(rn2^n/\eps)$ time to infidelity at most $\eps$?}

A concrete route to proving this, as already observed in \cite{pelecanos2025mixed}, is to improve the time complexity of the random purification channel. This also gives a route to answering the further question: Does there exist an algorithm for learning $n$-qubit $k$-sparse rank-$r$ mixed states in $O(nkr/\eps)$ time to infidelity at most $\eps$? We show in \cref{section:mixed-state} that the sample complexity scales with the support size $k$, could we additionally get the time complexity of the random purification to scale optimally with $k$ (or $r$)?

When $r=1$, which is when the input state is pure, our work resolves the prior question up to log factors, scaling with $k$. Optimizing the log factor dependence when $r=1$ is also an independent interesting question, since it seems it would require fundamentally different techniques as compared to our algorithm to remove the log factors entirely, or show a tight matching lower bound.

\question{Does there exist an algorithm for learning $n$-qubit $k$-sparse pure states in $\Tilde{O}(kn/\eps)$ time to infidelity at most $\eps$, using only single qubit Pauli measurements?}
The work of \cite{grewal2026nearly} answers this question in the case of $k=2^n$, by obtaining $\Tilde{O}(2^n/\eps)$ time complexity scaling for general pure states. We conjecture that some log factor dependence may be necessary in this case. The generalization of the above question to arbitrary mixed states has also been studied by \cite{acharya2025pauli1,acharya2025pauli2}, for both $k,r=2^n$. Obtaining (near) optimal time complexity bounds, along with scaling in $k$, remains open.

\question{Does there exist an algorithm for learning states which are sparse in a fixed (possibly unknown) Fock basis which is sample optimal and time efficient?}
The work of \cite{chen2026optimal} recently obtained optimal sample complexity bounds for learning general $m$-mode bosonic and fermionic Gaussian states. However, their algorithm is not known to be efficiently implementable. There has also been work that obtains algorithms which are sample and time efficient but not sample optimal \cite{bittel2025optimal,iosue2025higher}. We then ask the further question -- if a state is promised to be $k$-sparse in some unknown Gaussian rotated Fock basis, can we learn it with $\poly(m,k,1/\eps)$ samples and time for fermions, possibly with an energy constraint dependence for bosons?

\section{The Phase State Algorithm}

As promised, we first give a complete description and proof of the phase state algorithm for intuitive clarity; and since we will reuse several of the key ideas and proofs in the general pure state case. In order to show the correctness of this algorithm and analyze it, we first need to prove some lemmas we will use.

\begin{algorithm}
\caption{$k$-Sparse Phase State Tomography}
\label{alg:k_sparse_phase_tomography}
\begin{algorithmic}[1]
\State \textbf{Input:} Copies of an $n$-qubit $k$-sparse phase state $\ket{\psi}=\frac{1}{\sqrt{k}}\sum_{x\in S} e^{i\theta_x}\ket{x}$, the support size $|S|=k$, and error parameters $\eps,\delta > 0$
\State \textbf{Output:} An estimate $\ket{\hat{\psi}}=\frac{1}{\sqrt{k}}\sum_{x\in S} e^{i\hat{\theta}_x}\ket{x}$, with phases recovered up to a global phase

\State Using $\Theta(k\log (k/\delta))$ copies of $\ket{\psi}$, learn the support $S$ in the computational basis.
\State Let $\mathcal{F}=\{f_1,\dots,f_m\}$ be the set of functions from \cref{lemma:linear-connectivity}, where $m=O(\log(k/\delta))$, and each $f_i:\{0,1\}^n\to\{0,1\}^{\ell}$, with $\ell=\lceil \log(2k)\rceil$.
\State Initialize an empty graph $G=(S,E)$ and an empty table $\Delta$ of relative phase estimates.
\For{$i=1,\dots,m$}
    \State Take $\Theta(k\log^2(k)\log(k/\delta)/\eps)$ fresh copies of $\ket{\psi}$.
    \For{each copy}
        \State Append an ancilla register initialized to $\ket{0^{\ell}}$.
        \State Apply the XOR oracle $U_{f_i}:\ket{x}\ket{0^{\ell}}\mapsto\ket{x}\ket{f_i(x)}$.
        \State Measure the ancilla register, obtaining some outcome $y$, and place the remaining post-measurement state in some abstract indexed bucket $(i,y)$.
    \EndFor
    \For{each $y$ such that $|\{x\in S:f_i(x)=y\}|=2$}
        \State Search over table of all $x$ to find the two pre-images of $y$. 
        \State Let $\{x_1,x_2\}=\{x\in S:f_i(x)=y\}$.
        \State Use \cref{lemma:two_tom} on bucket $(i,y)$ with $(x_1,x_2)$ to estimate $\Delta(x_1,x_2)=\theta_{x_2}-\theta_{x_1}$.
        \State Add the edge $\{x_1,x_2\}$ to $E$ and set $\Delta(x_2,x_1):=-\Delta(x_1,x_2)$.
    \EndFor
\EndFor
\State Let $r\in S$ be the lexicographically first element of $S$, and set $\hat{\theta}_r=0$.
\State Carry out a BFS on $G$ starting from the vertex $r$, and do as follows for each visited vertex
\For{each $x\in S\setminus\{r\}$}
    \State Consider the path $r=v_0,v_1,\dots,v_t=x$ in $G$, which by \cref{lemma:linear-connectivity} has length at most $t=O(\log(k))$.
    \State Set $\hat{\theta}_x:=\sum_{j=0}^{t-1}\Delta(v_j,v_{j+1})$.
\EndFor
\State \Return $\ket{\hat{\psi}}=\frac{1}{\sqrt{k}}\sum_{x\in S} e^{i\hat{\theta}_x}\ket{x}$.
\end{algorithmic}
\end{algorithm}

The following lemma will be useful for the pure state case as well, so we state and prove it more generally than needed for learning phase states.

\begin{lemma}[Estimating phase difference between a superposition of two basis states]\label{lemma:two_tom}
Consider the quantum state
\[
\ket{\psi} \coloneq a\ket{x} + be^{i\phi}\ket{y},
\qquad x,y \in \{0,1\}^n,\ x\ne y, \,\, a, b \in \mathbb{R}.
\]
Given as input the two bitstrings $(x,y)$ and $O\left(\log(1/\delta)/a^2b^2\varepsilon\right)$ copies of $\ket{\psi}$, with probability at least $1-\delta$, \cref{alg:rel_phase_est} returns an estimate $\hat{\phi}$ such that $\abs{\phi-\hat{\phi}} \le \sqrt{\varepsilon}$.
\end{lemma}

\begin{proof}
\cref{alg:rel_phase_est} is shown to be correct as follows.

\begin{algorithm} 
\caption{Estimating the relative phase $\phi$}
\label{alg:rel_phase_est}
\begin{algorithmic}[1]
\State \textbf{Input:} Bitstrings $x,y \in \{0,1\}^n$ with $x \ne y$, access to copies of $\ket{\psi}=a\ket{x}+be^{i\phi}\ket{y}$, and error parameters $\varepsilon,\delta > 0$

\State \textbf{Output:} An estimate $\hat{\phi}$ satisfying $\abs{\phi-\hat{\phi}} \le \sqrt{\varepsilon}$ with probability at least $1-\delta$

\State Let $R \gets \{ i \in [n] : x_i \ne y_i \}$
\State Let $j \gets \min R$
\For{each fresh copy of $\ket{\psi}$}
    \State Apply, for every $i \in R \setminus \{j\}$, a CNOT with control qubit $j$ and target qubit $i$
    \State Measure every qubit except qubit $j$
    \State Use the resulting one-qubit sample on qubit $j$ as input to a single-qubit tomography procedure by calculating Pauli $X,Y$ expectation values on it
\EndFor
\State After tomography on these samples, return the phase estimate $\hat{\phi}$ 
\end{algorithmic}
\end{algorithm}

Let $R$ be the set of indices $i$ such that $x_i \ne y_i$, which can be found in $O(n)$ time given $x,y$.
Let $j$ be the lexicographically first index in $R$.
Fix one copy of $\ket{\psi}$.
Apply a CNOT from qubit $j$ to each qubit indexed by $R \setminus \{j\}$. We then get the state
\[
\ket{\psi_j} = a\ket{r_1,...,x_j,...r_n}+be^{i\phi}\ket{r_1,...,\overline{x_j},...r_n}
\]
(for some $r_i \in \{0,1\}$) and then measure all qubits except $j$.

Since $x$ and $y$ differ exactly on the coordinates in $R$, after applying these CNOTs the two basis states agree on every qubit other than $j$.
Therefore, after measuring all qubits except $j$, the remaining qubit is in the state
\[
\ket{\psi_j'}=a\ket{0}+be^{i\phi}\ket{1},
\]
up to a global phase depending on the value of $x_j$, from which we can decide the sign. Thus, we have reduced to estimating the relative phase of the single-qubit state $\ket{\psi_j'}$. Our claim follows from using a slightly modified single-qubit tomography algorithm on multiple copies of the state to get an estimate $\hat{\phi}$ of the relative phase to the desired accuracy. In particular, for the qubit $\ket{\psi_j'}$ , taking expectation values over Pauli $X$ and $Y$ measurements satisfy
\[
    \bra{\psi_j'}X\ket{\psi_j'} = 2ab\cos(\phi),
    \qquad
    \bra{\psi_j'}Y\ket{\psi_j'} = 2ab\sin(\phi).
\]
Estimating both expectations to additive error $O(ab\sqrt{\eps})$ and returning the angle of the resulting vector gives relative phase error at most $\sqrt{\eps}$, and we get concentration from Hoeffding's inequality for the stated number of copies.
\end{proof}

\noindent In order to state the next lemma, we first define the following object.

\begin{definition}[Image-pairs-graph]
Given a set $S \subseteq \{0,1\}^n$, and a set of functions $\mathcal{F}=\{f_1,\dots,f_t\}$ such that $f_i:\{0,1\}^n\to\{0,1\}^{l}$, we define the image-pairs-graph $G_{S,\mathcal{F}} = (S,E)$ as the graph with vertices labeled by elements in $S$ where there is an edge between any two distinct vertices $u,v \in S$ iff they are the only two elements which hash to the same value for some function $f_i$; in other words if the following is true.
\[
\exists f_i: (f_i(u)=f_i(v)) \land \abs{\{s \in S: f_i(s)=f_i(u)\}}=2 \,.
\]
\end{definition}
As described in the algorithm, this graph also represents the relative phases learned corresponding to the state with support $S$. There is an edge $(u,v)$ in this graph iff we can learn the relative phase between basis state $\ket{u},\ket{v}$ such that $u,v \in S$. Now, we can describe the set of functions we need.

\begin{lemma}[Diameter bound from random linear functions]\label{lemma:linear-connectivity}
Given a set $S\subseteq \{0,1\}^n$ with $\abs{S}=k\ge 2$, there exists a set of $m=O(\log(k/\delta))$ functions $\mathcal F=\{f_1,\dots,f_m\}$ over some choice of randomness, where
\[
f_i:\{0,1\}^n\to \{0,1\}^{\ell}, \qquad
\ell=\lceil \log(2k)\rceil,
\]
such that each $f_i$ is of the form $f_i(x)=A_i x$ over
$\mathbb F_2$, and the image-pairs graph $G_{S,\mathcal F}$ has diameter
$O(\log k)$ with probability at least $1-\delta$. Moreover, each $f_i$ can be implemented by a circuit with depth $O(\log(n))$ and size $O(n\log(k))$.
\end{lemma}

\begin{proof}
The proof idea is as follows. First, we choose an $f_i$ which is a random linear map (or drawn from a family of 3-universal hash functions), with the size of image a constant multiple of $k$. Since this is a universal family of hash functions, any fixed pair $u,v\in S$ collides with probability $\Theta(1/k)$. We then show that conditioned on this collision, with constant probability no third element of $S$ collides with them. This creates a random partial matching on $G_{S,\mathcal F}$. Why does this suffice? For some intuition, it is known that the union of three random matchings in a $k$-vertex graph is a vertex-expander with high probability, and thus has diameter $O(\log(|S|))$ \cite{goyal2009expanders}. Although, the union of partial matchings do not seem to satisfy expansion properties in general, we show in \cref{lemma:random-matching-radius} that with $O(\log(k))$ partial matchings, we can obtain a similar guarantee on the diameter of $G_{S,\mathcal F}$ with constant probability. 

We now show that each such $f_i$ defines a random partial matching. Let $\ell = \lceil \log(2k)\rceil$, and $R=2^\ell$, so
$2k\le R<4k$. Choose $m$ independent uniformly random matrices
$A_i\in \mathbb F_2^{\ell\times n}$, $1\le i\le m$, and let $f_i(x)=A_i x$. For each
$i$, let $G_i$ be the graph on $S$ whose edges correspond exactly to the
image-pair edges created by $f_i$ i.e. $(u,v) \in G_i \iff f_i(u)=f_i(v) \land \abs{\{s \in S: f_i(s) = f_i(u)\}}=2$. Then,
$G_{S,\mathcal F}=G_1\cup\cdots\cup G_m$. 

Now, we analyze the probability of a single edge $(u,v)$ being added to $G_i$. Consider distinct $u,v\in S$, and let $d=u+v\ne 0$. For a randomly chosen linear map $A$ \footnote{we note that the proof also carries through if $f$ is chosen from a 3-universal family of hash functions, which is sufficient in particular where we use conditional independence in equation (3)},
\begin{equation}
\Pr[A u=A v]=\Pr[A d=0]=1/R \,.    
\end{equation}
Consider an arbitrary $s\in S\setminus\{u,v\}$ and let $e=s+u$. Since this implies $e\notin \{0,d\}$,
the vectors $e$ and $d$ are linearly independent over $\mathbb F_2$.
Thus, for uniformly random $A$, the random variables $Ad$ and $Ae$ are
independent and uniformly distributed in $\mathbb F_2^\ell$. Consequently, $\Pr[As=Au\mid Au=Av]=\Pr[Ae=0\mid Ad=0]=1/R$. Hence, using a union bound over all such $s$,
\begin{equation}
\Pr[\exists s\in S\setminus\{u,v\}: A s=A u\mid A u=A v] \le (k-2)/R \le 1/2 \,.   
\end{equation}
since $R \ge 2k$. Thus, by Bayes' Theorem combined with equations (4) and (5), and using $R \le 4k$,
\begin{equation}
\Pr[(u,v) \in G_i] = \Pr[(Au = Av)] \cdot (1- \Pr[\exists s\in S\setminus\{u,v\}: A s=A u\mid A u=A v]) \ge \frac{1}{R} \cdot \frac{1}{2} > \frac{1}{8k} \,.   
\end{equation}
By taking a union of $m=O(\log(k/\delta))$ independent matchings in $G_{S,\mathcal F}$, where each edge is included independently with probability $> 1/8k$, applying \cref{lemma:random-matching-radius}, with $c=1/8$, shows that the diameter of $G_{S,\mathcal F}$ is bounded by $O(\log(k))$ with probability at least $1-\delta$. 

Finally, we note that since each $f_i$ is a random linear map over $\mathbb{F}_2$, the quantum circuit which implements it is just implementing a random parity or XORs via a tree of CNOTs over $n$ bits for $l$ times, so the circuit is of size $O(n\ell)=O(n\log(k))$ and depth $O(\log(n))$. 
\end{proof}

\begin{lemma}[Diameter bound from random partial matchings]\label{lemma:random-matching-radius}
Consider a vertex set $V$ with $k\ge 2$ vertices, and let $G_1,\ldots,G_m$ be
independent random matchings on $V$. Let $G=\bigcup_i G_i$. Let us assume that each edge per matching is uniform randomly chosen with probability at least $p \ge c/k$ for some constant $0<c<1$, i.e.,
\[\exists 0<c<1: \forall i: \forall u,v \in V: \Pr[(u,v)\in E(G_i)]\ge c/k \,,\]
where $E(G_i)$ is the edge set of $G_i$. If $m=\Omega(\log(k/\delta))$,  then with probability at least $1-\delta$, the graph
$G$ has diameter at most $O(\log k)$.
\end{lemma}
\begin{proof}
It is well-known that Erdös-Rényi random graphs with edge probability $p= \Omega(\log(n)/n)$ have diameter $O(\log(n))$ \cite{bollobas1981diameter,frieze2015introduction}. Here, we use a similar argument by analyzing the number of vertices reached from a fixed vertex across multiple independent matchings. Consider an arbitrary vertex $u \in V$, and let $B_r(u)$ be the set of vertices $S'$ can reach by a path of at most $r$ edges over the union of the first $r$ matchings, $r \in \mathbb{Z}^+$, with the $i$'th edge picked from the $i$'th matching, $1\le i \le r$. Let $\abs{B_r(u)}=x_r$, and consider the random variable $X_{r+1} \coloneq \abs{B_{r+1}(u)\setminus B_r(u)}$. Now, we just need to show that $X_r$ is sufficiently large at each step. Let us assume the first $r$ matchings are done, inducing the set $B_r(u)$. Now, the expected number of edges in the new matching $G_{r+1}$ going across the cut $D_r=(B_{r}(u), S \setminus B_{r}(u))$ is as follows.
\begin{equation}
\mathbb{E}[\abs{D_r}] = \mathbb{E}[X_{r+1}] = \sum_{v_i \in B_r(u)}\sum_{v_j \in S \setminus B_r(u)}\Pr[(v_i,v_j) \in E(G_{r+1})] \ge x_r(k-x_r)\frac{c}{k}
\end{equation}
Since $G_{r+1}$ is a matching, each of the edges correspond to a distinct new vertex not in $B_r(u)$ which can be reached by a path of length at most $r+1$, putting this element in $X_{r+1}$. Notice that $X_{r+1} \le  x_r$ since $G_{r+1}$ is a matching. Thus, applying the Paley-Zygmund inequality with $\mathbb{E}[X_{r+1}^2] \le \mathbb{E}[X_{r+1}]x_r$  (since $X_{r+1}^2 \le X_{r+1}x_r$), and $\theta=1/2$,
\begin{equation}
    \Pr\left[X_{r+1} \ge \frac{cx_r(k-x_r)}{2k}\right] \ge \Pr[X_{r+1} \ge \frac{1}{2}\mathbb{E}\left[X_{r+1}]\right]  \ge \left(1-\frac{1}{2}\right)^2\frac{\mathbb{E}[X_{r+1}]^2}{\mathbb{E}[X_{r+1}^2]} \ge \frac{cx_r(k-x_r)}{4kx_r} = \frac{c(k-x_r)}{4k} 
\end{equation}
Now, if we assume $x_r \le k/2$, then we get
\begin{equation}
    \Pr\left[X_{r+1} \ge \frac{cx_r}{4}\right] \ge \frac{c}{8}
\end{equation}
Thus, conditioned on the size of $B_r(u)$ to be $x_r$, with constant probability at least $c/8$, 
\begin{equation}
\frac{\abs{B_{r+1}(u)}}{\abs{B_r(u)}} \ge 1+\frac{c}{4} \,.
\end{equation}

If we take several such matchings, then we now have to show that with high probability, sufficiently many of the matchings expand the set of vertices by a multiplicative factor of at least $1+c/4$. More precisely, let $q=c/8$, and denote a matching successful if it expands the current set of reachable vertices by a factor of at least $1+c/4$. The previous equation shows that, conditioned on all previously successful matchings, each new matching is successful with probability at least $q$. Thus, the number of \textit{successful} matchings stochastically dominates the binomial random variable $\operatorname{Bin}(m,q)$. The number of such successful matchings needed to cover strictly more than $k/2$ vertices (since we assumed $x_r \le k/2$) is at most $\left\lceil \log_{(1+c/4)}(k/2)\right\rceil+1=O(\log(k))$ since $c$ is constant. Thus, by a Chernoff bound on $\operatorname{Bin}(m,q)$, if $m=\Omega(\log(k/\delta'))$, then with probability at least $1-\delta'$, a fixed vertex $S'$ is connected to more than $k/2$ vertices, each by path of length at most $O(\log(k))$.

Now, notice that the above analysis was done for an arbitrary fixed vertex $S'$. We can apply this to all vertices simultaneously, each with error probability $\delta'=\delta/k$. By a union bound, we get that with probability at least $1-\delta$, 
\begin{equation}
m = \Omega(\log(k^2/\delta))=\Omega(\log(k/\delta))   
\end{equation}
matchings are sufficient to ensure that the set $\abs{B_m(u)} > k/2$. So, each vertex $u \in V$ is connected to more than $k/2$ other vertices, each by a path of length at most $O(\log(k))$. To obtain a bound on the diameter of $G$, consider any two vertices $u,v\in V$. By the pigeonhole principle, the sets $B_m(u)$ and $B_m(v)$ must have a vertex in common since there are a total of $k$ vertices in $G$. Therefore, the distance between any two vertices via the path through their common vertex is at most $O(2\log(k))=O(\log(k))$, from which our claim follows. 
\end{proof}

\noindent We are now ready to prove the correctness of \cref{alg:k_sparse_phase_tomography}.
\begin{theorem}[Restatement of \cref{theorem:phase-main-theorem-intro}]
Given as input multiple copies of an unknown $n$-qubit $k$-sparse phase state $\ket{\psi}$, the support size $k$, and error parameters $\eps,\delta > 0$, there exists an algorithm (\cref{alg:k_sparse_phase_tomography}) that outputs a classical description of an estimate $\ket{\hat{\psi}}$ satisfying $|\bra{\psi}\hat{\psi}\rangle|^2 \ge 1 - \eps$ with probability at least $1-\delta$. The algorithm performs only single-copy measurements and requires $\Tilde{O}(k/\eps)$ samples and $\Tilde{O}(kn/\eps)$ time, when setting $\delta=1/n$, where the $\Tilde{O}$ suppresses logarithmic factors in $k$ and $n$.
\end{theorem}

\begin{proof}
By the error analysis in \cref{lemma:error-analysis}, to get the required guarantee on the fidelity of $\ket{\hat{\psi}}$, it suffices for us to show that we learn each phase relative to some fixed reference basis state $\ket{r}$ to additive error at most $\sqrt{\eps/2}$. We show this as follows, following the steps in \cref{alg:k_sparse_phase_tomography}.

First, by measuring $O(k\log(k/\delta))$ copies of $\ket{\psi}$ in the computational basis, we can learn the full support $S$ with probability at least $1-\delta/4$ by a standard coupon collector argument.

Let $m$ be the size of the set of functions $\mathcal{F}$. We take a fresh copy of the sparse phase state and apply some $f_i \in \mathcal{F}$ on it, preparing the following state.
\begin{equation}
    \ket{\psi}=\frac{1}{\sqrt{k}}\sum_{x\in S} e^{i\theta_x}\ket{x} \xrightarrow{f_i} \frac{1}{\sqrt{k}}\sum_{x\in S} e^{i\theta_x}\ket{x}\ket{f_i(x)}.
\end{equation}
We then measure the second register, with some outcome $y$. Let us assume for now there exists exactly two inputs $x_1,x_2 \in S$ such that $f_i(x_1)=f_i(x_2)=y$, else we discard the copy of the state and repeat with a fresh copy. Now, we have prepared the state, up to global phase,
\begin{equation}
    \ket{\psi'} = \frac{1}{\sqrt{2}}(\ket{x_1}+e^{i(\theta_{x_2}-\theta_{x_1})}\ket{x_2}) \,,
\end{equation}
which we assign to bucket $(i,y)$.

Let $q$ be the total number of buckets $(i,y)$ over all $i$ for which there are exactly two elements of $S$ mapping to $y$. Since each such bucket contains exactly two elements of $S$, for every fixed $i$ there are at most $k/2$ such buckets. Therefore $q\le km/2$.

Now, we set the accuracy needed to estimate each relative phase in the image-pairs-graph $G_{S,\mathcal{F}}$. Since, by \cref{lemma:random-matching-radius}, the path from some reference state $r$ to any other basis state has length at most $t = O(\log(k))$ with probability at least $1-\delta/4$, by triangle inequality over the path it suffices to learn each edge phase to additive error at most $\sqrt{\eps/2}/t$. By \cref{lemma:two_tom}, given bitstrings $x_1,x_2$ and $l=O(\log(1/\delta')/\eps')$ copies of $\ket{\psi'}$, we can learn the relative phase between $\ket{x_1}$ and $\ket{x_2}$ to additive error at most $\sqrt{\eps'}$ with probability at least $1-\delta'$. Thus, we set $\eps'=\eps/2t^2$ and $\delta'=\delta/4q$. So, $l=O(t^2\log(km/\delta)/\eps)$.

Now we have to show that we indeed obtain the required number of copies of the state in each bucket, for each $f_i$. By expectation bounds for the $l$-coupon collector \cite{xu2011generalized}, and applying a standard concentration bound via Chernoff and union bound, $O(kl\log\log(k)+k\log(k/\gamma))$ samples suffice to collect at least $l$ samples in every relevant bucket for this $f_i$ with probability at least $1-\gamma$.

We set $\gamma=\delta/4m$. Taking a union bound over all $m$ functions, with probability at least $1-\delta/4$, every relevant bucket receives at least $l$ copies. 
It remains to count the samples. For each $f_i$, the number of copies used is
\begin{align}
C &\le O(kl\log\log(k)+k\log(k/\gamma)) \\
  &\le O(kt^2\log(km/\delta)\log\log(k)/\eps+k\log(km/\delta)) 
  \le \Tilde{O}(k\log^2(k)\log(k/\delta)/\eps)
\end{align}
hiding $\log\log k$ factors. Since there are $m$ functions, the total sample complexity of this algorithm is $Cm = \Tilde{O}(k\log^2(k)\log^2(k/\delta)/\eps)$ to succeed with probability at least $1-\delta$, taking a union bound over error in learning the support, the required path length, the minimum bucket size, and the number of functions.

Finally, by \cref{lemma:linear-connectivity}, each function $f_i$ followed by measurements can be implemented by an $O(n\log(k))$ size circuit of depth $O(\log(n))$, so it takes $Cmn\log(k)=\Tilde{O}(kn\log^3(k)\log^2(k/\delta)/\eps)$ time to learn and output all the relative phases assuming all the samples are generated and processed sequentially as is standard. The time to do the BFS traversal and other classical post-processing is also bounded by the number of edges in the graph. So, the total running time of the algorithm is $\Tilde{O}(kn\log^3(k)\log^2(k/\delta)/\eps) = \Tilde{O}(kn/\eps)$, when setting $\delta=1/n$.
\end{proof}

\section{The Pure State Algorithm}\label{section:sparse-pure}
We now show how to extend the algorithm to arbitrary sparse pure states. Consider the input state
\[
\ket{\psi}=\sum_{x\in S}\alpha_x\ket{x}=\sum_{x\in S}\sqrt{p_x}e^{i\theta_x}\ket{x},
\qquad S\subseteq\{0,1\}^n,\quad \abs{S} \le k,
\]
where $p_x>0$ for $x\in S$ and $\sum_{x\in S}p_x=1$. The phase state algorithm is the special case where the the probabilities $p_x$ are all equal to $1/k$.

The pure state algorithm is a direct generalization of the phase state algorithm with two additional components.
\begin{enumerate}[label = \arabic*)]
    \item Instead of applying a coupon collector style sampling procedure to learn the support of the state, we use a generic classical distribution learning algorithm to learn the probability distribution of the state in the $Z$ basis.
    \item Using the standard technique of bucketing, we partition all the basis states into bins which are nearly uniformly distributed; we also select the highest probability basis state as the common reference state $r$, and place it in every bin. We then learn a phase graph having small diameter for the basis states in each bin similarly to the phase state algorithm, allocating samples per bin according to their probabilities.
\end{enumerate}
Since states with probability mass much smaller than $\eps/k$ do not matter for fidelity, we first learn a high-probability subset of the support and an approximation to the distribution on it. We then partition the retained basis states into $L$ different probability bins labeled as $A_j, 1\le j \le L$ such that $A_j = \{x: t_j \le q_x \le 2t_j\}$ where $t_j = 2^{-j}$. For the bin $A_j$, we set $H_j=A_j\cup\{r\}$ and build the image-pairs graph using random linear functions as in the phase state algorithm on $H_j$. Since every such graph is connected with small diameter and contains $r$, all phases have the same global reference state and form an overall connected graph. Learning the phase graph requires a slight modification since we may have multiple probability bins for any particular value of a hash function, but this is handled by projecting onto the space for a particular bin.

The accuracy for learning the phase graphs for each $A_j$ is determined by its corresponding probability level. If we let $\hat\mu_j \coloneq \sum_{x\in A_j}q_x$ and $D_j=O(\log(|H_j|))$ is the diameter bound for the phase graph on $H_j$, then we have to learn every edge phase used at level $j$ to accuracy roughly
\[
    \xi_j \approx \sqrt{\frac{\eps}{L\hat\mu_jD_j^2}},
\]
where $L=O(\log(k/\eps))$ bounds the number of retained probability levels. Consequently, every $x\in A_j$ contributes phase error at most $D_j\xi_j$, and the weighted contribution of the level is
\[
    \sum_{x\in A_j}p_x(\hat\theta_x-\theta_x)^2=O(\eps/L).
\]
Summing over levels gives the weighted phase error required for fidelity.

\begin{algorithm}
\caption{$k$-Sparse Pure State Tomography}\label{alg:k_sparse_pure_tomography}
\begin{algorithmic}[1]
\State \textbf{Input:} Copies of an $n$-qubit state $\ket{\psi}=\sum_{x\in S}\sqrt{p_x}e^{i\theta_x}\ket{x}$ with $|S| = k$, and error parameters $\eps,\delta>0$.
\State \textbf{Output:} An estimate $\ket{\hat\psi}=\sum_{x\in S'}\sqrt{q_x}e^{i\hat\theta_x}\ket{x}$.
\State Run the preprocessing subroutine of \cref{lemma:classical-preprocessing} with error parameter $\eps$ and failure probability $\delta/4$, using computational basis measurements to generate samples from $p$.
\State Obtain a retained support $S'$, a probability distribution $q$ supported on $S'$, and probability buckets $(A_j,t_j)_{j=1}^L$, where $L=O(\log(k/\eps))$.
\State Fix the reference element $r\in S'$ such that $q_r=\max_{x\in S'}q_x$, breaking ties lexicographically, and set $\hat\theta_r=0$.
\For{$j=1,\ldots,L$}
    \State Let $\hat\mu_j\gets \sum_{x\in A_j}q_x$ and $H_j\gets A_j\cup\{r\}$.
    \State Let $h_j\gets |H_j|$.
    \If{$h_j=1$}
        \State continue.
    \EndIf
    \State Set $m_j = O(\log(h_jL/\delta)),D_j= O(\log(h_j)),\ell_j = \lceil\log(2h_j)\rceil$ according to \cref{lemma:linear-connectivity,lemma:classical-preprocessing}
    
    \State Draw independent random linear maps $f_{j,1},\ldots,f_{j,m_j}:\{0,1\}^n\to\{0,1\}^{\ell_j}$.
    \State Set $\xi_j = O\left(\sqrt{\eps/L\hat\mu_jD_j^2}\right)$ according to \cref{theorem:pure-main-theorem}.
    \State Initialize an empty graph $G_j=(H_j,E_j)$ and an empty table $\Delta_j$ of edge phase estimates.
    \For{$i=1,\ldots,m_j$}
        \State Compute the two-preimage buckets of $f_{j,i}$ inside $H_j$:
        \Statex
        \[
        \mathcal P_{j,i}=
        \left\{
        \{x,y\}\subseteq H_j:
        x\ne y,\ f_{j,i}(x)=f_{j,i}(y),\
        \left|\{z\in H_j:f_{j,i}(z)=f_{j,i}(x)\}\right|=2
        \right\}.
        \]
        \State Take
        \[
            M_j=O\left(\frac{1}{t_j}\left(
            \frac{\log(h_jLm_j/\delta)}{\xi_j^2}
            +\log\frac{h_jLm_j}{\delta}
            \right)\right)
        \]
        \Statex fresh copies. For each copy, append the hash register, apply the XOR oracle for $f_{j,i}$, and measure the hash register. If the outcome corresponds to a pair $\{x,y\}\in\mathcal P_{j,i}$, project the data register onto $\operatorname{span}\{\ket{x},\ket{y}\}$ and store the post-measurement state if the projection succeeds.
        \For{each $\{x,y\}\in\mathcal P_{j,i}$}
            \State Let $\lambda_{x,y}=c_\lambda\min\{q_x,q_y\}/(q_x+q_y)$.
            \State Use \cref{lemma:two_tom} on the stored states for $\{x,y\}$ with parameter $\lambda_{x,y}$ to estimate $\Delta_j(x,y)=\theta_y-\theta_x$ to error at most $\xi_j$.
            \State Add the edge $\{x,y\}$ to $E_j$ and set $\Delta_j(y,x)\gets-\Delta_j(x,y)$.
        \EndFor
    \EndFor
    \State Run BFS in $G_j$ from $r$. For every $x\in A_j\setminus\{r\}$, set $\hat\theta_x$ to be the sum of the stored edge estimates along the BFS path from $r$ to $x$.
\EndFor
\State \Return $\ket{\hat\psi}=\sum_{x\in S'}\sqrt{q_x}e^{i\hat\theta_x}\ket{x}$.
\end{algorithmic}
\end{algorithm}

In order to explain and analyze \cref{alg:k_sparse_pure_tomography}, we need to prove one additional lemma which generalizes the coarse coupon collector technique used in \cref{alg:k_sparse_phase_tomography}, combining it with bucketing probabilities.

\begin{lemma}[Classical preprocessing for learning and bucketing probabilities]\label{lemma:classical-preprocessing}
Let $p$ be an unknown distribution on $S \subseteq \{0,1\}^n$ such that $\abs{S}=k$. Given sample access to $p$ and parameters $\eta,\delta>0$, there is a classical algorithm which uses $O(k\log(k/\delta)/\eta)$ samples and outputs a set $S'\subseteq\{0,1\}^n$, a probability distribution $q$ supported on $S'$, and labeled probability buckets $(A_j,t_j)_{j=1}^L$, where each $A_j \subseteq S'$, $0 \le t_j \le 1$, such that the number of buckets is bounded by $L=O(\log(k/\eta))$, and with probability at least $1-\delta$ the following hold.
\begin{enumerate}
    \item $\sum_x\sqrt{p_xq_x}\ge 1-\eta$.
    \item The nonempty sets among $A_1,\ldots,A_L$ partition $S'$.
    \item For every nonempty level $A_j$ and every $x\in A_j$,
    \[
        t_j\le q_x<2t_j .
    \]
    \item For every $x\in S'$, $q_x=\Theta(p_x)$. Consequently, if we let
    \[  \mu_j \coloneq \sum_{x\in A_j}p_x, \qquad
        \hat\mu_j \coloneq \sum_{x\in A_j}q_x
    \]
    then $\hat\mu_j=\Theta(\mu_j)$ for every nonempty level $A_j$.
    \item Every nonempty level has mass
    \[
        \hat\mu_j=\Omega(\eta/L)
        \qquad\text{and}\qquad
        \mu_j=\Omega(\eta/L).
    \]
\end{enumerate}
\end{lemma}

\begin{proof}
The proof follows from standard results on learning distributions in Hellinger distance and \emph{bucketing} distributions. We defer the complete proof to the Appendix in \cref{lemma:classical-preprocessing-app}.
\end{proof}

\noindent It only remains to prove the correctness of \cref{alg:k_sparse_pure_tomography}.

\begin{theorem}[Restatement of \cref{theorem:pure-main-theorem}]\label{theorem:pure-main}
Given as input multiple copies of an unknown $n$-qubit $k$-sparse pure state $\ket{\psi}$, the support size $k$, and error parameters $\eps,\delta > 0$, there exists an algorithm (\cref{alg:k_sparse_pure_tomography}) that outputs a classical description of an estimate $\ket{\hat{\psi}}$ satisfying $|\bra{\psi}\hat{\psi}\rangle|^2 \ge 1 - \eps$ with probability at least $1-\delta$. The algorithm performs only single-copy measurements and requires $\Tilde{O}(k/\eps)$ samples and $\Tilde{O}(kn/\eps)$ time, when setting $\delta=1/n$, where the $\Tilde{O}$ suppresses logarithmic factors in $k,n$ and $1/\eps$.
\end{theorem}

\begin{proof}
We first run the preprocessing subroutine from \cref{lemma:classical-preprocessing} with error parameter $\eps$ and failure probability $\delta/4$. With probability at least $1-\delta/4$, we obtain a retained support $S'$, a distribution $q$ supported on $S'$, and probability buckets $(A_j,t_j)_{j=1}^L$, where $L=O(\log(k/\eps))$, such that
\[
    \sum_x\sqrt{p_xq_x}\ge 1-\eps.
\]
The nonempty buckets partition $S'$, and for every nonempty bucket $A_j$,
\[
    t_j\le q_x<2t_j \quad (\forall x\in A_j),
    \qquad
    q_x=\Theta(p_x)\quad (\forall x\in S'),
\]
and
\[
    \hat\mu_j\coloneq \sum_{x\in A_j}q_x=\Theta(\mu_j),
    \qquad
    \mu_j\coloneq \sum_{x\in A_j}p_x,
    \qquad
    \hat\mu_j,\mu_j=\Omega(\eps/L).
\]
Now, we set $r\in S'$ such that $q_r=\max_{x\in S'}q_x$ to be the reference state, and fix its global phase $\theta_r=0$.

For each nonempty bucket $A_j$, set $H_j=A_j\cup\{r\}$ and $h_j=|H_j|$. If $h_j=1$, there is no phase to learn. Otherwise, since $r$ maximizes $q_x$, every $x\in H_j$ satisfies $q_x\ge t_j$. We would like to learn the relative phases corresponding to the states in each bucket to with failure probability at most $\delta/4L$. So, similarly to the phase state analysis, by \cref{lemma:linear-connectivity}, we pick $m_j=O(\log(h_jL/\delta))$ independent random linear maps $f_{j,1},\ldots,f_{j,m_j}:\{0,1\}^n\to\{0,1\}^{\ell_j}$ with $\ell_j=\lceil \log(2h_j)\rceil$, getting an image-pairs graph $G_j$ on $H_j$ of diameter at most $D_j=O(\log(h_j))$ with failure probability at most $\delta/(4L)$. Union bounding over the $L$ buckets, all these diameter bounds hold with probability at least $1-\delta/4$.

It remains to show that the edge phases in these graphs are learned to sufficient accuracy. Consider any bucket $A_j$, a hash function $f_{j,i}$, and an edge $e=\{x,y\}\in\mathcal P_{j,i}$ where $P_{j,i}$ is the relative phase graph for the bucket $A_j$ on the $i'$th hash function for this bucket. On any sample, after applying the hash and measuring the hash register, we may have basis states in superposition which do not belong to $A_j$. This is unlike the case of phase states since there we would only have 1 bucket. To fix this, we apply an additional projection onto $\operatorname{span}\{\ket{x},\ket{y}\}$ which succeeds with total probability exactly $p_x+p_y$ (and takes time $O(n)$ just by checking equality to $x,y$ with ancilla and CNOTs). We then obtain the state
\[
    \ket{\psi_{x,y}}
    =
    \frac{\sqrt{p_x}e^{i\theta_x}\ket{x}
    +\sqrt{p_y}e^{i\theta_y}\ket{y}}{\sqrt{p_x+p_y}} .
\]
Its smaller squared amplitude is  $\alpha_{x,y} \coloneq \frac{\min\{p_x,p_y\}}{p_x+p_y}$.
Since $q_x=\Theta(p_x)$ on $S'$, we have $\alpha_{x,y}=\Theta\!\left(\frac{\min\{q_x,q_y\}}{q_x+q_y}\right)$.
Thus, by choosing $c_\lambda$ sufficiently small according to \cref{lemma:classical-preprocessing}, the value
$\lambda_{x,y}  = c_\lambda\frac{\min\{q_x,q_y\}}{q_x+q_y}$ used by the algorithm is a valid lower bound on $\alpha_{x,y}$. Therefore \cref{lemma:two_tom} learns
$\abs{\theta_y-\theta_x}$ to error at most $\xi_j$ with probability at least $1-\delta/8Lm_jh_j$ using at most
\[
    R_{x,y}
    =
    \Theta\!\left(
        \frac{q_x+q_y}{\min\{q_x,q_y\}}
        \frac{\log(h_jLm_j/\delta)}{\xi_j^2}
    \right)
\]
copies, where
\[
    \xi_j
    =
    O\left(\sqrt{\frac{\eps}{L\hat\mu_jD_j^2}}\right).
\]
It remains to check that the algorithm produces this many stored copies from its original copies. After applying a hash function, in order to get the edge
$e=\{x,y\}$, we project onto the desired state with probability exactly $p_x+p_y$. Hence the expected number of stored copies is $M_j(p_x+p_y)$ where the number of copies after each hash is 
\[
    M_j
    \coloneq
    O\!\left[
        \frac{1}{t_j}\left(
        \frac{\log(h_jLm_j/\delta)}{\xi_j^2}
        +\log\frac{h_jLm_j}{\delta}
        \right)
    \right]
\]
original copies for this hash. Since $p_x+p_y=\Theta(q_x+q_y)$ and every vertex of $H_j=A_j\cup\{r\}$ has probability mass in $q$ which at least $t_j$, by \cref{lemma:classical-preprocessing}, we get
\[
    M_j(p_x+p_y)
    =
    \Omega\!\left(
        \frac{q_x+q_y}{t_j}
        \left(
        \frac{\log(h_jLm_j/\delta)}{\xi_j^2}
        +\log\frac{h_jLm_j}{\delta}
        \right)
    \right)
    \ge
    \Omega\!\left(
        R_{x,y}
        +
        \log\frac{h_jLm_j}{\delta}
    \right),
\]
where the last inequality uses $\min\{q_x,q_y\}\ge t_j$, which is sufficient samples. A Chernoff bound therefore implies that $e$
receives at least $R_{x,y}$ stored copies except with probability at most $\delta/(8Lm_jh_j)$. Union bounding over all buckets, hashes, and edges, every edge receives enough stored copies with probability at least $1-\delta/8$.

Now, we run \cref{lemma:two_tom} for every edge with failure probability $\delta/(8Lm_jh_j)$. Since each $\mathcal P_{j,i}$ is a matching of size at most $h_j/2$, a union bound over all $L$ buckets, $m_j$ hashes, and $h_j$ edges shows that every edge phase estimate in every graph $G_j$ has error at most $\xi_j$ with probability at least $1-\delta/8$.

Now, consider the BFS tree from $r$ in each graph $G_j$. For every $x\in A_j$, the path from $r$ to $x$ has length at most $D_j$, and hence $\abs{\hat\theta_x-\theta_x}\le D_j\xi_j$. The error of $r$ is zero, since it is the reference. Now, we calculate the total error. By the preprocessing guarantee, $q_x=\Theta(p_x)$ on $S'$, and hence
\[
    \sum_{x\in A_j}\sqrt{p_xq_x}=O\!\left(\sum_{x\in A_j}q_x\right)=O(\hat\mu_j).
\]
Thus,
\begin{align*}
    \sum_{x\in S'}\sqrt{p_xq_x}\,(\hat\theta_x-\theta_x)^2
    &=\sum_{j=1}^L\sum_{x\in A_j}\sqrt{p_xq_x}\,(\hat\theta_x-\theta_x)^2  \\
    &\le \sum_{j=1}^L O(\hat\mu_jD_j^2\xi_j^2) \\
    &\le O(\eps).
\end{align*}
Choosing appropriate constants, we obtain the bounds
\[
    \sum_x\sqrt{p_xq_x}\ge 1-\eps/4
    \qquad\text{and}\qquad
    \sum_{x\in S'}\sqrt{p_xq_x}\,(\hat\theta_x-\theta_x)^2\le \eps/4.
\]
By \cref{lemma:pure-error-analysis}, this implies that with probability at least $1-\delta$,
\[
    |\bra{\psi}\hat{\psi}\rangle|^2\ge 1-\eps .
\]

This shows correctness. We now briefly analyze sample and time complexity. The preprocessing step uses $\Tilde{O}(k/\eps)$ samples. For one hash in bucket $A_j$, the number of samples used, with high probability (setting $\delta=1/n$), is
\[
    M_j=\Tilde{O}\!\left(
        \frac{1}{t_j}\cdot\frac{1}{\xi_j^2}
    \right)
    =
    \Tilde{O}\!\left(
        \frac{1}{t_j}\left(\frac{L\hat\mu_jD_j^2}{\eps}\right)
    \right).
\]
Since every nonempty bucket has $\hat\mu_j=\Omega(\eps/L)$, the second term is absorbed into the first. Thus, including the $m_j$ hashes and hiding logarithmic factors, we require $\Tilde{O}\!\left(\frac{L\hat\mu_j}{t_j\eps}\right)$ samples for bucket $j$.
Because $q_x\in[t_j,2t_j)$ on $A_j$, we have $\hat\mu_j/t_j=\Theta(|A_j|)$. So, bucket $j$ uses $\Tilde{O}\!\left(\frac{L|A_j|}{\eps}\right)$ copies. Thus, summing over all the disjoint nonempty buckets gives total sample complexity $\Tilde{O}(k/\eps)$.

For the running time, each preprocessing sample is measured in the computational basis which costs $O(n)$. For a fixed bucket and hash, building the classical hash table over $H_j$ costs $\Tilde{O}(h_jn)$ time. Since the buckets are disjoint and the same reference is added to each nonempty bucket, the total table building cost over all buckets and hashes is $\Tilde{O}(kn)$, which is dominated by the sampling cost. Each sample uses one linear hash circuit from \cref{lemma:linear-connectivity}, one projection, and one single-copy measurement for the tomography, with the total cost being $\Tilde{O}(n)$. Hence the total running time is $\Tilde{O}(kn/\eps)$. We note that the log factors in our sample and time complexity analyses are probably not tight and could be improved with further optimizations; we leave this for future work.
\end{proof}

\section{The Mixed State Algorithm}
\label{section:mixed-state}
We now show how an application of the random purification channel technique \cite{tang2025conjugate, girardi2025random} allows us to learn sparse mixed states with near optimal sample complexity. The key observation is that any purification of $k$-sparse rank-$r$ mixed state results in a $kr$-sparse pure state. This is analogous to prior work on reducing mixed state tomography to pure state tomography \cite{pelecanos2025mixed}, we simply observe that a similar more fine-grained reduction holds, which scales with the sparsity of the state. We restate their theorem here, first proposed in the work of \cite{tang2025conjugate}.

\begin{theorem}[Random purification channel]
\label{thm:random-purification-channel}
There is a quantum channel $\Phi_{\mathrm{Purify}}^{d,r}$ such that given $m$ copies of a rank-$r$ mixed state $\rho \in \mathbb{C}^{d \times d}$,
\begin{equation}
\Phi_{\mathrm{Purify}}^{d,r}\!\left(\rho^{\otimes m}\right) =
\mathbb{E}_{|\rho\rangle}
\left[
|\rho\rangle\!\langle\rho|^{\otimes m}
\right],   
\end{equation}
where the expectation is over a uniform random purification \footnote{this refers to a distribution over purifications which is obtained by applying a Haar-random unitary to the second register of any fixed purification $\ket{\rho} \in \mathbb{C}^d \otimes \mathbb{C}^r$} $|\rho\rangle \in \mathbb{C}^d \otimes \mathbb{C}^r$ of $\rho$. This channel $\Phi_{\mathrm{Purify}}^{d,r}$ can moreover be implemented within $\eps$ error in diamond distance  by a circuit which takes time $\operatorname{poly}\!\left(m,\log(d),\log(1/\eps)\right).$
\end{theorem}

Now, for states which have rank-$r$ and are also $k$-sparse, we obtain the following corollary. Recall that $k$-sparse mixed states are those states which have support of size $k$ when measured in the computational basis.

\begin{corollary}[Random purification channel for sparse states]
\label{cor:random-purification-channel-sparse}
There is a quantum channel $\Phi_{\mathrm{Purify}}^{d,r}$ such that given $m$ copies of a rank-$r$ $k$-sparse mixed state $\rho \in \mathbb{C}^{d \times d}$,
\begin{equation}
\Phi_{\mathrm{Purify}}^{d,r}\!\left(\rho^{\otimes m}\right) =
\mathbb{E}_{|\rho\rangle}
\left[
|\rho\rangle\!\langle\rho|^{\otimes m}
\right],   
\end{equation}
where the expectation is over a uniform random purification $\ket{\rho} \in \mathbb{C}^d \otimes \mathbb{C}^r$ of $\rho$, such that it is is in a subspace isomorphic to $\mathbb{C}^k \otimes \mathbb{C}^r$ . Moreover, $\ket{\rho}$ is $kr$-sparse in the computational basis. This channel $\Phi_{\mathrm{Purify}}^{d,r}$ can moreover be implemented within $\eps$ error in diamond distance  by a circuit which takes time $\operatorname{poly}\!\left(m,\log d,\log(1/\eps)\right).$
\end{corollary}
\begin{proof}
It suffices to show that every purification of a $k$-sparse rank-$r$ mixed state is supported on a $kr$-dimensional subspace, so for a random purification this will also be true. Let $S \subseteq [d]$ denote the support of $\rho$ in the computational basis, with $|S|=k$. Then, since $\rho = \Pi_S \rho \Pi_S$, and by definition of a purification, $\Tr_2[\ket{\rho}\bra{\rho}] = \rho$, this implies $(\bra{i} \otimes I)\ket{\rho}=0$ for any $i \notin S$. Thus,
any purification $|\rho\rangle \in \mathbb{C}^d \otimes \mathbb{C}^r$ of $\rho$ must lie in
\[
\operatorname{span}\{|i\rangle : i\in S\} \otimes \mathbb{C}^r.
\]
where the first space is isomorphic to $\mathbb{C}^k$. So, $|\rho\rangle$ lies in a subspace isomorphic to $\mathbb{C}^k \otimes \mathbb{C}^r$, and in particular is a $kr$-sparse pure state over the computational basis. The sample and runtime implementation guarantees follow directly from \cref{thm:random-purification-channel}.
\end{proof}

This brings us to our main theorem for learning $k$-sparse rank-$r$ mixed states. We follow the generic reduction algorithm as described in \cite{pelecanos2025mixed}.

\begin{algorithm}
\caption{$k$-Sparse Mixed State Tomography}\label{alg:k_sparse_mixed_tomography}
\begin{algorithmic}[1]
\State \textbf{Input:} Copies of an unknown $n$-qubit $k$-sparse rank-$r$ mixed state $\rho$, the sparsity bound $k$, the rank bound $r$, and error parameters $\eps,\delta>0$.
\State \textbf{Output:} An estimate $\hat\rho$.
\State Apply the random purification channel $\Phi_{\mathrm{Purify}}^{d,r}$ to $m$ copies of $\rho$, obtaining $m$ copies of a random purification $\ket{\rho} \in \mathbb{C}^d \otimes \mathbb{C}^r$.
\State Run \cref{alg:k_sparse_pure_tomography} on the $m$ purified samples with sparsity parameter $kr$, error parameter $\eps$, and failure probability $\delta/2$, obtaining an estimate $\ket{\hat\rho}$.
\State Return a succinct or sparse classical description of $\hat{\rho} = \Tr_2(\ket{\hat\rho}\bra{\hat{\rho}})$
\end{algorithmic}
\end{algorithm}

\begin{theorem}[Restatement of \cref{theorem:mixed-main-theorem}]\label{theorem:mixed-main}
Given as input multiple copies of an unknown $n$-qubit $k$-sparse rank-$r$ mixed state $\rho$, the support size $k$, the rank $r$, and error parameters $\eps,\delta > 0$, there exists an algorithm (\cref{alg:k_sparse_mixed_tomography}) that outputs a classical description of an estimate $\hat{\rho}$ satisfying $F(\rho,\hat{\rho}) \ge 1-\eps$ with probability at least $1-\delta$. The algorithm requires $m=\Tilde{O}(kr/\eps)$ samples and $\poly(m,n)$ time, when setting $\delta=1/n$, where the $\Tilde{O}$ suppresses logarithmic factors in $k$ and $1/\eps$.
\end{theorem}

\begin{proof}
In step 3 of \cref{alg:k_sparse_mixed_tomography}, we obtain $m$ purified copies $\ket{\rho}$ of $\rho$, which we know by \cref{cor:random-purification-channel-sparse} is $kr$-sparse. By \cref{theorem:pure-main}, setting error parameters $\eps, \delta/2$ and sparsity parameter $kr$, we learn an estimate $\ket{\hat{\rho}}$ such that
\[
1-|\langle \rho|\hat{\rho}\rangle|^2 \leq \eps
\]
with probability at least $1-\delta/2$.
Since $\ket{\rho}$ and $\ket{\hat\rho}$ are purifications of $\rho$ and
$\hat\rho$, respectively, by Uhlmann's theorem,
\[
F(\rho,\hat\rho)
\geq
|\langle\rho|\hat\rho\rangle|^2
\geq
1-\eps.
\]
Since this guarantee holds for every fixed purification, it also holds
after averaging over the random purification.
The additional $\delta/2$ error in diamond distance from implementing the purification channel gives total failure probability at most $\delta$. Note that applying \cref{theorem:pure-main} with sparsity $kr$ and error $\eps$ requires $m=\Tilde{O}(kr/\eps)$ samples and $\Tilde{O}(krn/\eps)$ time. However, the overall runtime of our mixed state tomography algorithm is dominated by the random purification channel, which takes $\poly(m,n)$ time for $\delta=1/n$. This completes the proof.
\end{proof}

\section*{Acknowledgments}
I thank Scott Aaronson for teaching the graduate course that inspired the basis of this work, and generally for helpful discussions and comments. I also thank Ronak Ramachandran for valuable feedback on an earlier draft of this paper, as well as Vishnu Iyer and William Kretschmer for helpful discussions.

I made limited use of GPT 5.5 and 5.6 Pro during the research process, in particular to help with the analysis in the proof of \cref{theorem:pure-main} and to write up initial drafts of the proofs in the Appendix. I have verified and edited the proofs for correctness and clarity.

\bibliographystyle{alpha}

\bibliography{references}

\vspace{1em}

\appendix

\begin{center}
    {\LARGE \textbf{Appendix}}
\end{center}

\section{Lower Bound}
\begin{proposition}[Lower bound for learning $k$-sparse rank-$r$ mixed states]\label{lemma:lower-bound}
Let $\rho \in \mathbb{C}^{d \times d}$ be an arbitrary $k$-sparse rank-$r$ mixed state such that $k > 1$, and $1 \le r \le k \le d$. Then, $m=\Omega(kr/\eps)$ samples are necessary to learn $\rho$ to fidelity greater than $1-\eps$. This lower bound holds even if the computational basis support of $\rho$ is fixed in advance and revealed to the algorithm.
\end{proposition}
\begin{proof}
We reduce arbitrary rank-$r$ mixed-state tomography in dimension $k$
to $k$-sparse rank-$r$ mixed-state tomography. Suppose, for the sake of contradiction, that there exists an algorithm $\mathcal{A}$ which learns any $k$-sparse rank-$r$ mixed state to fidelity at least $1-\eps$ using $m=o\!\left(\frac{kr}{\eps}\right)$ samples. We assume $\mathcal{A}$ is given the computational basis support as input as well.

\noindent Now, let $\rho\in\mathbb{C}^{k\times k}$ be an arbitrary rank-$r$
mixed state. Consider an isometry
\[
V:\mathbb{C}^k\rightarrow\mathbb{C}^{d},
\]
such that $d\geq k$, which maps the standard basis vectors of $\mathbb{C}^k$ to $k$ fixed computational-basis vectors.  Applying this isometry, we obtain the state $\sigma=V\rho V^\dagger$, which has rank $r$ and is supported on at most $k$ vectors in the computational basis, and hence is $k$-sparse.

We now apply our sparse state tomography algorithm $\mathcal{A}$ with $m$ copies of $\sigma$, obtaining an estimate $\widehat{\sigma}$ such
that their fidelity $F(\sigma,\widehat{\sigma})\geq 1-\eps$. We now want to convert this back to a fidelity guarantee for learning $\rho$, so we invert the isometry $V$ on its image.

So, let $\Pi=VV^\dagger$ and $ q=\Tr(\Pi\widehat{\sigma})$. On the success event, $q>0$. Then $\widehat{\rho} = \frac{V^\dagger\widehat{\sigma}V}{q}$.  Since $\sigma$ is supported on the image of $V$, and by the definition of fidelity, 
\[
F(\rho,\widehat{\rho}) = F(\sigma,\widehat{\sigma})/q. 
\] 
As $q\leq 1$, it follows that
\[ 
F(\rho,\widehat{\rho}) \geq F(\sigma,\widehat{\sigma}) \geq 1-\eps. 
\] 
Thus, algorithm $\mathcal{A}$ would be able to learn an arbitrary rank-$r$ mixed state in dimension $k$ to infidelity $\eps$ using 
$m=o\!\left(kr/\eps\right)$ samples. However, rank-$r$ mixed-state tomography in dimension $k$ to infidelity $\eps$ requires 
$\Omega\!\left(kr/\eps\right)$ samples \cite{yuen2023improved, scharnhorst2025optimal}. Therefore, $\mathcal{A}$ requires at least 
\[ 
m=\Omega\!\left(\frac{kr}{\eps}\right). 
\] 
samples of $\rho$, thus proving our lower bound.
\end{proof}

\section{Error Analysis}
\subsection{Phase States}
\begin{lemma}[Phase states error analysis]\label{lemma:error-analysis}
Let
\[
\ket{\psi}=\frac{1}{\sqrt{k}}\sum_{x\in S} e^{i\theta_x}\ket{x}
\qquad\text{and}\qquad
\ket{\hat{\psi}}=\frac{1}{\sqrt{k}}\sum_{x\in S} e^{i\hat{\theta}_x}\ket{x},
\]
where $S\subseteq\{0,1\}^n$ and $|S|=k$. Suppose that, after fixing a global phase convention,
\[
\gamma=\max_{x\in S}\abs{\Delta\hat{\theta}_x}\le \sqrt{\eps/2},
\qquad
\text{where}
\Delta\hat{\theta}_x=\hat{\theta}_x-\theta_x .
\]
Then
\[
\abs{\bra{\psi}\hat{\psi}\rangle}^2\ge 1-\eps .
\]
\end{lemma}
\begin{proof}
The fidelity between the two states is as follows.
\begin{align}
   \abs{\bra{\psi}\hat\psi\rangle}^2 
   &= \left| \sum_{x\in S} \frac{1}{k} e^{i(\hat{\theta}_x - \theta_x)} \right|^2 \\
   &= \frac{1}{k^2}\sum_{x,y \in S}e^{i((\hat{\theta}_x - \theta_x) - (\hat{\theta}_y - \theta_y))} \\
   &= \frac{1}{k^2}\sum_{x,y \in S}[\cos((\hat{\theta}_x - \theta_x) - (\hat{\theta}_y - \theta_y)) + i\sin((\hat{\theta}_x - \theta_x) - (\hat{\theta}_y - \theta_y))] \\
   &= \frac{1}{k^2}\sum_{x,y \in S}
   \cos\!\big((\hat{\theta}_x - \theta_x) - (\hat{\theta}_y - \theta_y)\big) \qquad \text{(since $\sin(\theta)+\sin(-\theta)=0$)}\\ 
   &= \frac{1}{k^2}\sum_{x,y \in S} \cos\!\big(\Delta\hat{\theta}_{x} - \Delta\hat\theta_{y}\big) \\
   &=  \frac{1}{k^2} \sum_{x,y \in S} \left[1 - \Bigl(1-\cos\!\big(\Delta\hat{\theta}_{x} - \Delta\hat\theta_{y}\big)\Bigr)\right] \\
   &= 1 - \frac{1}{k^2}\sum_{x,y \in S} \Bigl(1-\cos\!\big(\Delta\hat{\theta}_{x} - \Delta\hat\theta_{y}\big)\Bigr) \qquad \text{(since $\sum_{x,y \in S}1 = k^2$)} \\
   &\ge 1 - \frac{1}{k^2}\sum_{x,y \in S} \frac{\big(\Delta\hat{\theta}_{x} - \Delta\hat\theta_{y}\big)^2}{2} 
   \qquad \text{(since $1-\cos(t)\le t^2/2$)} \\
   &\ge 1 - \frac{1}{k^2}\sum_{x,y \in S} \frac{(2\gamma)^2}{2} 
   \qquad \text{(by the triangle inequality)}\\
   &\ge 1 - \frac{1}{k^2}\sum_{x,y \in S}\eps  \\
   &= 1 - \eps .
\end{align}
Note that all these estimates are up to some global phase convention, relative to some fixed reference basis state.
\end{proof}

\subsection{Pure States}
\begin{lemma}[Pure states error analysis]\label{lemma:pure-error-analysis}
Let
\[
\ket{\psi}=\sum_x\sqrt{p_x}e^{i\theta_x}\ket{x}
\qquad\text{and}\qquad
\ket{\hat\psi}=\sum_x\sqrt{q_x}e^{i\hat\theta_x}\ket{x}.
\]
Suppose that $\sum_x\sqrt{p_xq_x}\ge 1-\eta$ with $\eta \le 1/2$, and that after fixing a global phase convention, we write each phase error 
$\Delta_x=\hat\theta_x-\theta_x$ such that
\[
\sum_{x:q_x>0}\sqrt{p_xq_x}\,\Delta_x^2\le \eta.
\]
Then
\[
|\bra{\psi}\hat\psi\rangle|^2\ge 1-4\eta.
\]
\end{lemma}
\begin{proof}
After multiplying $\ket{\hat\psi}$ by a global phase, recall that the errors $\Delta_x$ are relative to the same reference phase. Then
\begin{align}
\operatorname{Re}(\bra{\psi}\hat\psi\rangle)
&= \sum_x\sqrt{p_xq_x}\cos(\Delta_x)\\
&\ge \sum_x\sqrt{p_xq_x}-\frac{1}{2}\sum_x\sqrt{p_xq_x}\Delta_x^2
\qquad\text{(since $1-\cos(t)\le t^2/2$)}\\
&\ge 1-\eta-\eta/2\ge 1-2\eta.
\end{align}
Therefore $|\bra{\psi}\hat\psi\rangle|\ge 1-2\eta$, and hence
\[
|\bra{\psi}\hat\psi\rangle|^2\ge (1-2\eta)^2\ge 1-4\eta,
\]
where in the last step we used $\eta\le 1/2$.
\end{proof}

\section{Classical preprocessing}
\begin{lemma}[Restatement of \cref{lemma:classical-preprocessing}]\label{lemma:classical-preprocessing-app}
Let $p$ be an unknown distribution on $S \subseteq \{0,1\}^n$ such that $\abs{S}=k$. Given sample access to $p$ and parameters $\eta,\delta>0$, there is a classical algorithm which uses $O(k\log(k/\delta)/\eta)$ samples and outputs a set $S'\subseteq\{0,1\}^n$, a probability distribution $q$ supported on $S'$, and labeled probability buckets $(A_j,t_j)_{j=1}^L$, where each $A_j \subseteq S'$, $0 \le t_j \le 1$, such that the number of buckets is bounded by $L=O(\log(k/\eta))$, and with probability at least $1-\delta$ the following hold.
\begin{enumerate}
    \item $\sum_x\sqrt{p_xq_x}\ge 1-\eta$.
    \item The nonempty sets among $A_1,\ldots,A_L$ partition $S'$.
    \item For every nonempty level $A_j$ and every $x\in A_j$,
    \[
        t_j\le q_x<2t_j .
    \]
    \item For every $x\in S'$, $q_x=\Theta(p_x)$. Consequently, if we let
    \[  \mu_j \coloneq \sum_{x\in A_j}p_x, \qquad
        \hat\mu_j \coloneq \sum_{x\in A_j}q_x
    \]
    then $\hat\mu_j=\Theta(\mu_j)$ for every nonempty level $A_j$.
    \item Every nonempty level has mass
    \[
        \hat\mu_j=\Omega(\eta/L)
        \qquad\text{and}\qquad
        \mu_j=\Omega(\eta/L).
    \]
\end{enumerate}
\end{lemma}

\begin{proof}
Draw $N=O\!\left(\frac{k\log(k/\delta)}{\eta}\right)$ samples from $p$, and let $\bar p$ be the empirical distribution. By the standard empirical Hellinger learning guarantee \cite{canonne2020short}, this allows us to learn the distribution $\bar p$ such that
\[
    \sum_x\sqrt{p_x\bar p_x}\ge 1-\eta
\]
with probability at least $1-\delta/3$.

\noindent Now, set
\[
    \rho=\frac{\eta}{100k},
    \qquad
    T=\{x:\bar p_x\ge 2\rho\}.
\]
By multiplicative Chernoff bounds and a union bound over the support of $p$, with probability at least $1-\delta/3$ the following hold simultaneously:
\[
    p_x\ge 4\rho \Longrightarrow x\in T,
    \qquad
    x\in T \Longrightarrow p_x\ge \rho,
    \qquad
    \bar p_x=\Theta(p_x)\quad\text{for all }x\in T .
\]
In particular,
\[
    p(T^c)=O(k\rho)=O(\eta),
    \qquad
    \bar p(T^c)=O(k\rho)=O(\eta).
\]

Let $\tilde q$ be $\bar p$ restricted to $T$ and renormalized. Then $\tilde q_x=\Theta(p_x)$ for every $x\in T$. Moreover, removing $T^c$ and renormalizing changes the error by only $O(\eta)$, so we still have
\[
    \sum_x\sqrt{p_x\tilde q_x}\ge 1-O(\eta).
\]

Now, we bucket $T$ according to the values of $\tilde q_x$, following the standard bucketing construction used in ~\cite{batu2001testing}. Thus each bucket is a dyadic level for $\tilde q$, and since $\tilde q_x=\Theta(p_x)$ on $T$, the true probabilities are within a constant factor of the bucketed probabilities on each level. Since every retained point has $\tilde q_x=\Omega(\eta/k)$, the number of possible nonempty dyadic levels is
\[
    L=O(\log(k/\eta)).
\]

Discard every level whose $\tilde q$-mass is less than $\eta/(100L)$. Since there are at most $L$ levels, the total discarded $\tilde q$-mass is $O(\eta)$. Since $\tilde q_x=\Theta(p_x)$ on $T$, the total discarded true mass is also $O(\eta)$.

Let $S'$ be the union of the remaining levels, and let $q$ be $\tilde q$ restricted to $S'$ and renormalized. This final normalization changes all remaining probabilities by a factor $1+O(\eta)$, so $q_x=\Theta(p_x)$ for every $x\in S'$. The same truncation and renormalization argument, followed by scaling $\eta$, gives
\[
    \sum_x\sqrt{p_xq_x}\ge 1-\eta.
\]

Finally, after adjusting the dyadic thresholds by the common normalization factor, each nonempty level $A_j$ satisfies
\[
    t_j\le q_x<2t_j
    \qquad\text{for all }x\in A_j.
\]
Each remaining level has $q$-mass at least $\Omega(\eta/L)$ by construction, and since $q_x=\Theta(p_x)$ on $S'$, the mass from $p$ is also $\Omega(\eta/L)$. This proves the lemma.
\end{proof}

\end{document}